\documentclass{amsart}
\usepackage{graphicx,cite}
\newtheorem{thm}{Theorem}[section]
\newtheorem{cor}[thm]{Corollary}

\newtheorem{prop}[thm]{Proposition}
\theoremstyle{definition}

\theoremstyle{remark}
\newtheorem{rem}[thm]{Remark}
\numberwithin{equation}{section}

\newcommand{\be}{\begin{equation}}
\newcommand{\ee}{\end{equation}}
\newcommand{\bea}{\begin{eqnarray}}
\newcommand{\eea}{\end{eqnarray}}
\begin{document}

\title[On Certain Triple $q$-Integral Equations]{ Certain  triple $q$-integral equations involving third Jackson $q$-Bessel functions as kernel}%
\author{Z.S. Mansour and M.A. AL-Towailb}%
\address{Z.S.I. Mansour, Department of Mathematics, King Saud University, Riyadh, KSA}%
\email{zsmansour@ksu.edu.sa }%
\address{M.A. AL-Towailb, Department of Natural and Engineering Science, King Saud University, Riyadh, KSA}%
\email{mtowaileb@ksu.edu.sa}%

\thanks{This research is supported by the DSFP program
of the King Saud University in Riyadh through grant DSFP/MATH 01 and by National plan of Science and Technology project number 14-MAT623-02.}%


\begin{abstract}
In this paper, we employ the fractional $q$-calculus in solving a  triple  system of $q$-Integral equations,
where the kernel is the third Jackson $q$-Bessel functions. The solution  is reduced to  two simultaneous Fredholm $q$-integral equation of the second kind. Examples are included. We also apply a result in~\cite{O.A} for solutions of dual $q^2$-integral equations to solve  certain triple integral equations.
\end{abstract}
\maketitle
\section{\bf{Introduction}}
\noindent Some three-parts mixed boundary value problems of the mathematical theory of elasticity are solved by reducing them to triple integral equations. Many of the triple integral equations are of the form
\begin{eqnarray*}
 \int_0^\infty  A(u)K(u,x)\,du&=& f(x), \quad 0 < x < a,\\
\int_0^\infty  w(u)A(u)K(u,x)\,du&=& g(x),\quad  a < x < b,\\
\int_0^\infty  A(u)K(u,x)\,du&=& h(x), \quad  b < x < \infty,
\end{eqnarray*}
 where $w(u)$ is the weight function, $K(u,x)$ is the kernel function. Several authors have described various methods to solve dual and triple integral equations especially when the kernel is a Bessel function. For the dual integral equations,  see for example~\cite{Love,Kes,ES,Cop,Nas,Nob,Nob2,Pet}. For the triple integral equations, see for  example~\cite{A.P,J.C.1,J.C.2,J.C.3,J.C.4,J.S,B.M,C.J,N.A,K.N}.
 In this paper, we consider triple $q$-integral equation where the kernel is the third Jackson
$q$-Bessel function and the $q$-integral is Jackson
$q$-integral. It is worth mentioning that  different approaches for solving dual $q$-integral equation is in~\cite{O.A}. Also, solutions for dual and triple sequence involving $q$-orthogonal polynomials is in~\cite{AMZ}.  This paper is organized as follows.  The next section is introductory section includes the main notions and notations we need in our investigations.   In Section 3, we solve the triple $q$-integral
equations by reducing the system to two simultaneous Fredholm
$q$-integral equation of the second kind, we shall use a method due
to Singh, Rokne and Dhaliwal \cite{B.M}. The approach depends on    fractional $q$-calculus. We include solutions of two dual $q$-integral equations as special cases of the solution of the triple $q$-integral equation included in this section, and we show that this coincides with the results in~\cite{O.A}.  In the last section, we use a result from~\cite{O.A} for a solution of dual $q^2$-integral equations to solve triple $q^2$-integral equations. The result of this section is a $q$-analogue  of the results introduced by cooke in~\cite{J.C.1}.

\section{\bf{q-Notations and Results}}\label{q-not.}
In this paper, we assume that $q$ is a positive number less than one. We introduce some of the needed $q$-notations and results (see \cite{M.H}).\vskip1mm

Let $t> 0$, $A_{q,t}$, $B_{q,t}$ and $\mathbb{R}_{q,t,+}$ be the sets defined by $$ A_{q,t}:=\{tq^n: n\in \mathbb{N}_0\}, \quad B_{q,t}:=\{tq^{-n}: n\in \mathbb{N}\},$$
$$ \mathbb{R}_{q,t,+}:=\{tq^k: k\in \mathbb{Z}\},$$ where $\mathbb{N}_0:=\{0,1,2,...\}$, and $\mathbb{N}:=\{1,2,...\}$. (Note that if $t=1$, we write $A_{q}$, $B_{q}$ and $\mathbb{R}_{q,+}$).
We follow Gasper and Rahman~\cite{G.M} for the definitions of the $q$-shifted factorial, multiple $q$-shifted factorials, basic hypergeometric series, Jackson $q$-integrals, the $q$-gamma and beta functions. We also follow Annaby and Mansour~\cite{M.H} for the definition of the $q$-derivative at zero.

\noindent Let $\alpha\in \mathbb{C}$, we will use the following notation
\begin{equation*}
\left[\begin{array}{c}
  \alpha \\
  k
\end{array}\right]_{q}= \left\{
                          \begin{array}{ll}
                            1, &  k=0; \\
                            \frac{(1-q^\alpha)(1-q^{\alpha-1})... (1-q^{\alpha-k+1})}{(q;q)_k}, &  k\in \mathbb{N}.
                          \end{array}
                        \right.
\end{equation*}

\noindent For $\eta\in \mathbb{C}$ and a function $f$ defined on $\mathbb{R}_{q,+}$, we define the following spaces
\begin{equation*}
L_{q,\eta}(\mathbb{R}_{q,+}):= \Big\{ f: \|f\|_{q,\eta}:= \int_0^\infty |t^\eta f(t)| d_qt < \infty \Big\},
\end{equation*}
\begin{equation*}
L_{q,\eta}(A_{q}):= \Big\{ f: \|f\|_{A_q,\eta}:= \int_0^1 |t^\eta f(t)| d_qt < \infty \Big\},
\end{equation*}
and
\begin{equation*}
L_{q,\eta}(B_{q}):= \Big\{ f: \|f\|_{B_q,\eta}:= \int_1^\infty |t^\eta f(t)| d_qt < \infty \Big\}.
\end{equation*}
Note that $ L_{q,\eta}(\mathbb{R}_{q,+})= L_{q,\eta}(A_{q})\cap L_{q,\eta}(B_{q}). $\vskip1mm

\noindent The third Jackson $q$-Bessel function $J_\nu^{(3)}(z;q)$, see \cite{M.P} and \cite{F.H}, is defined by
\begin{align}\label{Jackson}
J_\nu(z;q)=J_\nu^{(3)}(z;q):&=
\frac{(q^{\nu+1};q)_\infty}{(q;q)_\infty}z^{\nu} {_1\phi_1}
(0;q^{\nu+1};q,qz^2)\\ &=
\frac{(q^{\nu+1};q)_\infty}{(q;q)_\infty}\sum_{n=0}^\infty (-1)^n
\frac{q^{n(n+1)/2}z^{2n+\nu}}{(q;q)_n(q^{\nu+1};q)_n}, z\in
\mathbb{C},\end{align}
and satisfies the following relations (see
\cite{R.F}):
\begin{equation}\label{a}
D_q\left[(.)^{-\nu}J_\nu(.;q^2)\right](z)=- \frac{q^{1-\nu}z^{-\nu}}{1-q} J_{\nu+1}(qz;q^2),
\end{equation}
\begin{equation}\label{b}
D_q\left[(.)^{\nu}J_\nu(.;q^2)\right](z)=\frac{z^{\nu}}{1-q} J_{\nu-1}(z;q^2).
\end{equation}
Also, for $\Re(\nu) >-1$, the $q$-Bessel function $J_\nu(.;q^2)$ satisfies (see \cite{H.T}):

\begin{equation}\label{c}
\left| J_\nu(q^n;q^2)\right|\leq \frac{(-q^{2};q^2)_{\infty}(-q^{2\nu+2};q^2)_\infty}{(q^{2};q^2)_{\infty}} \left\{
                                                                                                          \begin{array}{ll}
                                                                                                            q^{n\nu}, & \mbox{ if }  n\geq 0; \\
                                                                                                            q^{n^2-(\nu+1)n},& \mbox{ if } n< 0.
                                                                                                          \end{array}
                                                                                                        \right.
\end{equation}
 We recall that the functions $\cos(z; q)$ and $\sin(z; q)$ are defined for $z\in \mathbb{C}$ by
$$\cos(z; q) := \frac{(q^2; q^2)_\infty}{(q; q^2)_\infty}\, (zq^{-\frac{1}{2}}(1-q))^\frac{1}{2}\, J_{-\frac{1}{2}}(z(1-q)/\sqrt{q};q^2),$$
$$\sin(z; q) := \frac{(q^2; q^2)_\infty}{(q; q^2)_\infty}\, (z(1-q))^\frac{1}{2}\, J_{\frac{1}{2}}(z(1-q);q^2).$$

 We need the following results from \cite{M.H}:
\begin{prop}\label{cor}
Let $\alpha$, $\beta \in \mathbb{C}$, $\rho,t\in \mathbb{R}_{q,+}$. If $\mathfrak{R}(\beta)> \mathfrak{R}(\alpha)>-1$, the
\[
\begin{gathered}
 \int_0^\infty t^{\alpha -\beta+1} J_\alpha (\xi t; q^2)J_\beta (\rho t; q^2)\,d_q t  \\
= \left\{
 \begin{array}{ll}
   0, &   \xi > \rho; \\
   \frac{(1-q)(1-q^2)^{1-\beta+\alpha}}{\Gamma_{q^2}(\beta-\alpha)}\xi^{\alpha}\rho^{(\beta-2\alpha-2)}(q^2\xi^2/\rho^2;q^2)_{\beta-\alpha-1},&  \xi\leq \rho.
\end{array}
  \right.
 \end{gathered}
 \]
\end{prop}
\begin{prop}\label{int}
Let $\nu$ and $\alpha$ be complex numbers such that $\Re(\nu)>-1$.
Then for $\rho, u \in \mathbb{R}_{q,+}$
\begin{equation*}
\int_\rho^\infty x^{2\alpha-\nu-1} ({\rho^2}/{x^2}; q^2)_{\alpha-1}
J_\nu (ux; q^2) \,d_q x  = \frac{(1-q)\Gamma_{q^2}(\alpha)}{(1-q^2)^{1-\alpha}}\rho^{\alpha-\nu}u^{-\alpha}q^{\alpha}
 J_{\nu-\alpha}(u
\rho/q; q^2).
 \end{equation*}
\end{prop}

\bigskip

\begin{prop}\label{prop.1}
Let $x$, $\nu$ and $\gamma$ be complex numbers and $u\in \mathbb{R}_{q,+}$.
Then, for $\Re(\gamma)>-1$ and $\Re(\nu)>-1$ the following identity  holds
\begin{equation}\label{8}
 \begin{gathered}
 \int_0^x  \rho^{\nu+1} ({q^2\rho^2}/{x^2}; q^2)_{\gamma} J_\nu (u\rho; q^2) \,d_q \rho =\\ x^{\nu-\gamma+1} u^{-\gamma-1} (1-q)(1-q^2)^\gamma \Gamma_{q^2}(\gamma + 1)J_{\gamma+\nu +1}(u x; q^2).
 \end{gathered}
 \end{equation}
Moreover, if $\Re(\gamma)>0$ and $\Re(\nu)>-1$, then
 \begin{equation}\label{9}
 \begin{gathered}
 \int_x^\infty \rho^{2\gamma-\nu-1} ({x^2}/{\rho^2}; q^2)_{\gamma-1} J_\nu (u\rho; q^2) \,d_q \rho \\ = x^{\gamma-\nu} u^{-\gamma} (1-q)q^\gamma \frac{(q^2;q^2)_\infty}{(q^{2\gamma};q^2)_\infty}J_{\nu-\gamma}(\frac{u x}{q}; q^2).
 \end{gathered}
 \end{equation}

 \end{prop}

\begin{cor}\label{cor:1}
Let $x, u$ and $\alpha$ be complex numbers such that $u\in \mathbb{R}_{q,+}$, $\Re(\alpha)>-1$ and $\Re(\nu)>-1$.  Then
\begin{equation}\label{20}
\begin{gathered}
 u^\alpha J_{\nu-\alpha}(ux; q^2)= \\ \frac{(1-q^2)^{\alpha}}{\Gamma_{q^2}(1-\alpha )} x^{\alpha-\nu-1} D_{q,x}\left[x^{-2\alpha} \int_0^x  \rho^{\nu+1} ({q^2\rho^2}/{x^2}; q^2)_{-\alpha} J_\nu (u\rho; q^2)\,d_q \rho\right].
  \end{gathered}\end{equation}
\end{cor}

\begin{proof}
 Applying (\ref{8}) with $\gamma = -\alpha$, we have:
 \begin{equation}\label{99}
 \begin{gathered}
 \int_0^x  \rho^{\nu+1} ({q^2\rho^2}/{x^2}; q^2)_{-\alpha} J_\nu (u\rho; q^2) \,d_q \rho \\ = x^{\nu+\alpha+1} u^{\alpha-1} (1-q)(1-q^2)^{-\alpha} \Gamma_{q^2}(1-\alpha )J_{\nu-\alpha +1}(u x; q^2).
 \end{gathered}
 \end{equation}
 Multiply both sides of equation (\ref{99}) by $x^{-2\alpha}$, and then calculate the $q$-derivative of the two sides with respect to $x$ and using (\ref{b}), we get the  required result.
\end{proof}

\begin{cor}
Let $x, u$ and $\alpha$ be complex numbers such that $u\in \mathbb{R}_{q,+}$, $\Re(\alpha)>0$ and $\Re(\nu)>-1$. Then
\begin{equation}\label{int2}
\begin{gathered}
 u^\alpha J_{\nu+\alpha}(ux; q^2)= \\-\frac{(1-q^2)^{\alpha}q^{2\alpha+\nu-2}x^{\alpha+\nu-1}}{\Gamma_{q^2}(1-\alpha )} D_{q,x} \int_x^\infty \rho^{-2\alpha-\nu+1} ({x^2}/{\rho^2}; q^2)_{-\alpha} J_\nu (u\rho; q^2) \,d_q \rho.
 \end{gathered}
 \end{equation}
\end{cor}
\begin{proof}
The proof is similar to the proof of Corollary~\ref{cor:1} and is omitted.
\end{proof}

\vskip .5 cm

Koornwinder and Swarttouw~\cite{H.T} introduced the following inverse pair of $q$-Hankel integral transforms under the side condition $f,\,g\in L_{q}^2(\mathbb{R}_{q,+})$:
\be\label{HT}
g(\lambda)=\int_{0}^{\infty}xf(x)J_{\nu}(\lambda\,x;q^2)\,d_qx;\quad f(x)=\int_{0}^{\infty}\lambda f(\lambda)J_{\nu}(\lambda\,x;q^2)\,d_q\lambda,
\ee
where $\lambda,\,x\in\mathbb{R}_{q,+}$.

In the following, we introduce
A $q$-analogue of the Riemann-Liouville fractional integral operator is introduced in \cite{W.A} by Al-Salam through
\begin{equation*}
I^{\alpha}_q f(x):= \frac{x^{\alpha-1}}{\Gamma_q(\alpha)}\int_0^x (qt/x;q)_{\alpha-1} f(t)\,d_qt,
\end{equation*}
$\alpha \not \in \{-1,-2,...\}$. In \cite{R.P}, Agarwal defined the $q$-fractional derivative to be
\begin{equation*}
D_q^{\alpha} f(x):= I^{-\alpha}_q f(x)= \frac{x^{-\alpha-1}}{\Gamma_q(-\alpha)}\int_0^x (qt/x;q)_{-\alpha-1} f(t)\,d_qt.
\end{equation*}
We shall also use that
\be\label{Abel}
I_q^{\alpha}D_q^\alpha f(x)=f(x)-I_q^{1-\alpha}f(0)\frac{x^{\alpha-1}}{\Gamma_q(\alpha)},\; 0<\alpha<1.
\ee
see~\cite[Lemma 4.17]{M.H}.

In the following, we introduce  some $q$-fractional operators  that we use in solving the triple $q$-integral equations under consideration. The technique of using fractional operators in solving dual and triple integral equations is not new. See for example~\cite{Agarwal,Pet,O.A}.
 In~\cite{W.A}, Al-Salam defined a  two parameter $q$-fractional operator by
\begin{equation*}\label{ws2}
K_{q}^{\eta,\alpha}\phi(x):=\frac{q^{-\eta}x^{\eta}}{\Gamma_q(\alpha)}\int_{x}^{\infty}\big(x/t;q\big)_{\alpha-1}t^{-\eta-1}\phi(tq^{1-\alpha})\,d_qt,
\end{equation*}
$\alpha\neq -1,-2,\ldots$. This is a $q$-analogue  of the
Erd\'{e}lyi and Sneddon fractional operator, cf.~\cite{AE,ES},
\[
K^{\eta,\alpha}f(x)=\frac{x^{\eta}}{\Gamma(\alpha)}\int_{x}^{\infty}\big(t-x)^{\alpha-1}t^{-\eta-1}f(t)\,
dt.
\]
In~\cite{O.A}, the authors introduced  a slight modification of the operator
$K_q^{\eta,\alpha}$. This operator is denoted by
${\mathcal{K}_q^{\eta,\alpha}}$ and defined by
\begin{equation}\label{kkk}\begin{split}
\mathcal{K}_q^{\eta,\alpha}\phi(x):=&\frac{q^{-\eta}x^{\eta}}{{\Gamma}_q(\alpha)}\int_x^\infty(x/t;q)_{\alpha-1}t^{-\eta-1}\phi(qt)d_qt,
\end{split}\end{equation}
where $ \alpha\neq -1,-2,\ldots$.
In case of $\eta=-\alpha$,  we set
\begin{equation}\label{K}
\begin{gathered}
\begin{split}\mathcal{K}_q^{\alpha}f(x)&:=q^{-\alpha}x^{\alpha}q^{\frac{\alpha(\alpha-1)}{2}}\mathcal{K}_q^{-\alpha,\alpha}f(x)\\
&=\frac{q^{\frac{-\alpha(\alpha-1)}{2}}}{\Gamma_q(\alpha)}\int_{x}^{\infty}t^{\alpha-1}(x/t;q)_{\alpha-1}f(qt)\,d_qt.
\end{split}
\end{gathered}
\end{equation}
This is a slight modification of the operator $K^\alpha f(x;q)$ introduced in ~\cite[(19.4.8)]{M.P} and by Al-salam in~\cite{W.A}. Note that this operator satisfies the following semigroup identity
\begin{equation}\label{prop.k}
\mathcal{K}_q^{\alpha} \mathcal{K}_q^{\beta} \phi(x)= \mathcal{K}_q^{\alpha+\beta} \phi(x), \quad \mbox { for all } \alpha \mbox { and } \beta.
\end{equation}
The proof of \eqref{prop.k} is completely similar to the proof of~\cite[Theorem 5.13]{M.H} and is omitted.

\begin{prop}
Let $\alpha \in \mathbb{C}, x\in B_q$. If $\Phi\in L_{q,\alpha -1}(B_q)$ and $G(x)=D_{q,x} \mathcal{K}_q^{\alpha}\Phi(x)$, then $$ \Phi(x)=-q^{\alpha-1 \dot{\dot{\dot{}}}} \mathcal{K}_q^{1-\alpha}G(\frac{x}{q}).$$
\end{prop}

\begin{proof}
First, we show that $$G(x)= -q^{1-\alpha} \mathcal{K}_q^{(\alpha-1)}\Phi(qx).$$ According to (\ref{K}), we have
\be\begin{split}\label{I}
G(x)&= \frac{q^{-\alpha(\alpha-1)/2}}{\Gamma_q(\alpha)}D_{q,x} \int_x^\infty t^{\alpha-1}(x/t;q)_{\alpha-1}\Phi(qt)\,d_qt \cr &=\frac{q^{-\alpha(\alpha-1)/2}}{x(1-q) \Gamma_q(\alpha)} \Big[\int_x^\infty t^{\alpha-1}(x/t;q)_{\alpha-1}\Phi(qt)\,d_qt- \int_{qx}^\infty t^{\alpha-1}(qx/t;q)_{\alpha-1}\Phi(qt)\,d_qt\Big].
\end{split}\ee
Note that $$ \int_{qx}^\infty g(t)\,d_qt= \int_{x}^\infty g(t)\,d_qt + x(1-q)g(x),$$
 so, (\ref{I}) can be written  as
$$ G(x)= \frac{q^{\frac{-\alpha(\alpha-1)}{2}}}{\Gamma_q(\alpha)}\Big[\int_x^\infty t^{\alpha-1}\Big(D_{q,x}(x/t;q)_{\alpha-1}\Big)\Phi(qt)\,d_qt - x^{\alpha-1}(q;q)_{\alpha-1}\Phi(qx)\Big].$$
 But
\[
D_{q,x}(x/t;q)_{\alpha-1}= -\frac{(1-q^{\alpha-1})}{t(1-q)}(qx/t;q)_{\alpha-2}= -\frac{1}{t}[\alpha-1](qx/t;q)_{\alpha-2}.
\]
Hence,
\begin{eqnarray*}
G(x)&=&- \frac{q^{-\alpha(\alpha-1)/2}}{\Gamma_q(\alpha)}[\alpha-1]\int_x^\infty t^{\alpha-2} (qx/t;q)_{\alpha-2} \Phi(qt)\,d_qt - x^{\alpha-1}(q;q)_{\alpha-1}\Phi(qx)\cr &=& - \frac{q^{-\alpha(\alpha-1)/2}}{\Gamma_q(\alpha)}[\alpha-1]\int_{qx}^\infty t^{\alpha-2} (qx/t;q)_{\alpha-2} \Phi(qt)\,d_qt \cr &=& - \frac{q^{-\alpha(\alpha-1)/2}}{\Gamma_q(\alpha-1)}\int_{qx}^\infty t^{\alpha-2} (qx/t;q)_{\alpha-2} \Phi(qt)\,d_qt = - q^{1-\alpha}  \mathcal{K}_q^{(\alpha-1)}\Phi(qx)
\end{eqnarray*}
This implies, $$  \mathcal{K}_q^{(\alpha-1)}\Phi(x)= -q^{\alpha-1}G(x/q), $$
 and by using (\ref{prop.k}), we obtain the result and completes the proof.
\end{proof}

\section{{\bf A system of triple $q$-Integral Equations}}
The goal of this section is to solve the following triple $q$-integral equations:

\begin{eqnarray}\label{F:1}
 \int_0^\infty  \psi (u) J_\nu (u\rho; q^2) \,d_q u&=& f_1 (\rho), \quad  \rho\in A_{q,a},\\\label{F:2}
   \int_0^\infty  u^{-2\alpha} \psi(u)\left[1+w(u)\right] J_\nu (u\rho; q^2)\,d_q u&=& f_2 (\rho), \quad   \rho\in A_{q,b}\cap B_{q,a},\\\label{F:3}
   \int_0^\infty  \psi (u) J_\nu (u\rho; q^2) \,d_q u&=& f_3 (\rho), \quad   \rho\in B_{q,b},
\end{eqnarray}
where $0<a<b<\infty$, and  $\alpha$, $\nu$  are  complex numbers satisfying \[\Re(\nu) > -1,\quad\mbox{and}\quad 0< \Re(\alpha)<1. \]   $\psi$ is an unknown function to be determined, and $f_i $ $(i=1,2,3)$ are known functions, and $w$ is a non-negative bounded function defined on $\mathbb{R}_{q,+}$.

\vskip .5 cm
Clearly from~\eqref{c}, a sufficient condition for the convergence of the $q$-integrals on the left hand side of \eqref{F:1}--\eqref{F:2} is
that
\be \label{SC}\psi\in L_{q,\nu}(\mathbb{R}_{q,+})\cap L_{q,\nu-2\alpha}(\mathbb{R}_{q,+}).\ee
   For getting the solution of the triple $q$-integral equations \eqref{F:1}--\eqref{F:3}, we define a function $C$  by
   \[C(u):= u^{-2\alpha}\psi (u)\left[1+w(u)\right],\quad u\in\mathbb{R}_{q,+}.\]
      Hence,
 $$\psi(u)= u^{2\alpha} C(u)- u^{2\alpha} C(u)\left[\frac{w(u)}{1+w(u)}\right],$$ and the triple $q$-integral equation (\ref{F:1})--(\ref{F:3}) can be represented as:
 \begin{equation}\label{1}
 \begin{gathered}
 \int_0^\infty  u^{2\alpha} C(u) J_\nu (u\rho; q^2) \,d_q u- \int_0^\infty  u^{2\alpha} C(u)\left[\frac{w(u)}{1+w(u)}\right] J_\nu (u\rho; q^2) \,d_q u \\
 = f_1 (\rho), \quad\quad   \rho\in A_{q,a}
 \end{gathered}
 \end{equation}

 \begin{equation}\label{2}
 \int_0^\infty  C(u) J_\nu (u\rho; q^2)\,d_q u= f_2 (\rho), \quad   \rho\in A_{q,b}\cap B_{q,a}
 \end{equation}

\begin{equation}\label{3}
\begin{gathered}
 \int_0^\infty u^{2\alpha} C(u) J_\nu (u\rho; q^2) \,d_q u- \int_0^\infty  \dfrac{w(u)}{1+w(u)}u^{2\alpha}C(u) J_\nu (u\rho; q^2) \,d_q u \\
 = f_3 (\rho), \quad   \rho\in B_{q,b}
 \end{gathered}
 \end{equation}
 Now assume that $C:=C_1+C_2$ such that
 \[\int_0^\infty  C_1(u) J_\nu (u\rho; q^2) \,d_q u= g_1 (\rho), \;  \rho\in A_{q,b},\]
 \[\int_0^\infty  C_2(u) J_\nu (u\rho; q^2) \,d_q u= g_2 (\rho), \,\rho\in B_{q,a},\]
  where \[ f_1(\rho)=g_1(\rho)+g_2(\rho),\quad \rho\in A_{q,b}\cap B_{q,a}.\] So, the triple $q$-integral equations (\ref{1})--(\ref{3}) can be rewritten in the following form:
 \begin{equation}\label{4}
 \begin{gathered}
 \int_0^\infty  u^{2\alpha}\left[C_1(u)+C_2(u)\right] J_\nu (u\rho; q^2) \,d_q u- \\ \int_0^\infty u^{2\alpha}\left[C_1(u)+C_2(u)\right]\dfrac{w(u)}{1+w(u)} J_\nu (u\rho; q^2) \,d_q u= f_1 (\rho), \quad  \rho\in A_{q,a},
 \end{gathered}
 \end{equation}
 \begin{equation}\label{5}
 \int_0^\infty  C_1(u) J_\nu (u\rho; q^2) \,d_q u= g_1 (\rho), \quad   \rho\in A_{q,b},
 \end{equation}
 \begin{equation}\label{6}
 \int_0^\infty  C_2(u) J_\nu (u\rho; q^2) \,d_q u= g_2 (\rho), \quad   \rho\in B_{q,a},
 \end{equation}
 \begin{equation}\label{7}
 \begin{gathered}
 \int_0^\infty  u^{2\alpha}\left[C_1(u)+C_2(u)\right] J_\nu (u\rho; q^2) \,d_q u - \\ \int_0^\infty  \dfrac{w(u)}{1+w(u)}u^{2\alpha}\left[C_1(u)+C_2(u)\right] J_\nu (u\rho; q^2) \,d_q u = f_3 (\rho), \quad   \rho\in B_{q,b}.
 \end{gathered}
 \end{equation}

 \begin{prop}\label{prop:1}
 Let the functions $\psi_1$, $\psi_2$ be defined by
 \begin{equation}\label{13}
\psi_1(x):=\int_0^\infty u^\alpha C_1(u) J_{\nu-\alpha}(ux; q^2)\,d_qu,\quad x\in B_{q,b},
\end{equation}
\begin{equation}\label{131}
\psi_2(x):=\int_0^\infty u^\alpha C_2(u) J_{\nu+\alpha}(ux; q^2)\,d_qu, \quad   x \in A_{q,a},                                                                                                        \end{equation}
provided that $0<\Re \alpha<1$, $\Re \nu>-1$, $\Re(\nu+\alpha)>0$ and
$C_1\in L_{q,\nu}(\mathbb{R}_{q,+})$, $C_2\in L_{q,-t}(\mathbb{R}_{q,+})$ where
\[\Re \nu+2>\Re t>-\Re \nu+2\Re(1-\alpha).\]
Then for $u\in\mathbb{R}_{q,+},$
  \begin{equation}\label{141}
C_1(u)= u^{1-\alpha} \left[\int_0^b x \Phi_1(x)J_{\nu-\alpha}(ux; q^2)\,d_qx + \int_b^\infty  x \Psi_1(x)J_{\nu-\alpha}(ux; q^2)\,d_qx\right],
\end{equation}
\begin{equation}\label{14}
C_2(u)= u^{1-\alpha} \left[\int_0^a  x \Psi_2(x)J_{\nu+\alpha}(ux; q^2)\,d_qx+ \int_a^\infty  x \Phi_2(x)J_{\nu+\alpha}(ux; q^2)\,d_qx\right] ,
\end{equation}
where
\begin{equation}\begin{split}\label{phi1}\Phi_1(x)&= \frac{(1-q^2)^{\alpha}x^{\alpha-\nu-1}}{\Gamma_{q^2}(1-\alpha )} D_{q,x}\left[x^{-2\alpha} \int_0^x
g_1 (\rho)  \rho^{\nu+1} ({q^2\rho^2}/{x^2}; q^2)_{-\alpha} \,d_q\rho\right]\\
&=(1-q^2)^{\alpha}x^{\alpha-\nu-1}D_{q^2,x}^{\alpha}\left(t^{\nu/2}g_1(\sqrt{t})\right)(x),\;x\in  A_{q,b},
\end{split}\end{equation}
\begin{equation}\label{phi2}\begin{split}\Phi_2(x)&= -\frac{(1-q^2)^{\alpha}q^{2\alpha+\nu-2}x^{\alpha+\nu-1}}{\Gamma_{q^2}(1-\alpha )} D_{q,x} \int_x^\infty g_2(\rho) \rho^{1-2\alpha-\nu} ({x^2}/{\rho^2}; q^2)_{-\alpha}  \,d_q \rho, \\
&=-q^{\frac{\alpha(1-\alpha)}{2}}(1-q^2)^{\alpha}x^{\alpha+\nu-1}D_{q^2,x}\mathcal{K}_{q^2}^{(1-\alpha)}\left[t^{-\nu/2}g_2(\sqrt{t})\right](\frac{x}{q^2}),\;x\in B_{q,a}.
\end{split}
 \end{equation}
\end{prop}

\begin{proof}
 We start with proving \eqref{phi1}. Let $x\in A_{q,b}$.  Multiply both sides of (\ref{5}) by $x^{-2\alpha} \rho^{\nu+1} ({q^2\rho^2}/{x^2}; q^2)_{-\alpha}$ and integrate with respect to $\rho$ from $0$ to $x$, we get
\begin{equation}\label{110}
\begin{gathered}
\int_0^x x^{-2\alpha} \rho^{\nu+1} ({q^2\rho^2}/{x^2}; q^2)_{-\alpha}  \int_0^\infty  C_1(u) J_\nu (u\rho; q^2) \,d_q u\,d_q \rho = \\ \int_0^x g_1 (\rho) x^{-2\alpha} \rho^{\nu+1} ({q^2\rho^2}/{x^2}; q^2)_{-\alpha} \,d_q\rho.
\end{gathered}
 \end{equation}
We can prove that the double $q$-integral on the left hand side of \eqref{110} is absolutely convergent for $0<\Re (\alpha)<1$ and for $\Re(\nu)>-1$ provided that $C_1\in L_{q,\nu}(\mathbb{R}_{q,+})$. So, we can interchange the order of the $q$-integrations to obtain
\begin{equation}\label{111}
\begin{gathered}
\int_0^\infty C_1(u) x^{-2\alpha} \int_0^x \rho^{\nu+1} (\frac{q^2\rho^2}{x^2}; q^2)_{-\alpha} J_\nu (u\rho; q^2) \,d_q \rho\, d_q u = \\ \int_0^x g_1 (\rho) x^{-2\alpha} \rho^{\nu+1} (\frac{q^2\rho^2}{x^2}; q^2)_{-\alpha}\, d_q\rho.
\end{gathered}
 \end{equation}
Calculate the $q$-derivative of the two sides of (\ref{111}) with respect to $x$ and using (\ref{20}), we get
\begin{equation}\label{121}
\int_0^\infty u^\alpha C_1(u) J_{\nu-\alpha}(ux; q^2)\,d_qu = \Phi_1(x), \quad x\in A_{q,b},
\end{equation}
 where $$\Phi_1(x)= \frac{(1-q^2)^{\alpha}x^{\alpha-\nu-1}}{\Gamma_{q^2}(1-\alpha )} D_{q,x}\left[x^{-2\alpha} \int_0^x
g_1 (\rho) x^{-2\alpha} \rho^{\nu+1} (\frac{q^2\rho^2}{x^2}; q^2)_{-\alpha} \,d_q\rho\right].$$

Now, we prove \eqref{phi2}. Let $x\in B_{q,a}$. Multiply both sides of (\ref{6}) by  $\rho^{-2\alpha-\nu+1}({x^2}/{\rho^2}; q^2)_{-\alpha}$ and $q$-integrate with respect to $\rho$ from $x$ to $\infty$, we get
\begin{equation}\label{10}
\begin{gathered}
\int_x^\infty \rho^{-2\alpha-\nu+1}({x^2}/{\rho^2} ; q^2)_{-\alpha} \int_0^\infty  C_2(u) J_\nu (u\rho; q^2) \,d_q u \,d_q \rho =\\\int_x^\infty g_2 (\rho) \rho^{-2\alpha-\nu+1} ({x^2}/{\rho^2}; q^2)_{-\alpha} \,d_q\rho.
 \end{gathered}
 \end{equation}
From (\ref{c}), we can prove that $u^t\,J_{\nu}(u;q^2)$ is bounded on $R_{q,+}$ provided that $\Re(t+\nu)>-1$. So, if we take $t$ such that   $\Re \nu+2>\Re t> -\Re \nu+2\Re(1-\alpha)$, we can prove that the double $q$-integral \[\int_x^\infty \rho^{1-2\alpha-\nu} ({x^2}/{\rho^2}; q^2)_{-\alpha} \int_0^\infty  C_2(u) J_\nu (u\rho; q^2) \,d_q u \,d_q \rho \]
  is absolutely convergent and we can interchange the order of the $q$-integration to obtain
\begin{equation}\label{11}
\begin{gathered}
\int_0^\infty C_2(u) \int_x^\infty   \rho^{1-2\alpha-\nu} ({x^2}/{\rho^2}; q^2)_{-\alpha} J_\nu (u\rho; q^2) \,d_q \rho\, d_q u \\ = \int_x^\infty g_2 (\rho) \rho^{-2\alpha-\nu+1} ({x^2}/{\rho^2}; q^2)_{-\alpha} \,d_q\rho.
 \end{gathered}
 \end{equation}
Calculating the $q$-derivative of the  two sides of (\ref{11}) with respect to $x$ and using \eqref{int2} yields
\begin{equation}\label{12}
\int_0^\infty u^\alpha C_2(u) J_{\nu+\alpha}(ux; q^2)\,d_qu =\Phi_2(x), \quad x\in B_{q,a},
\end{equation}
 where $$ \Phi_2(x)= -\frac{(1-q^2)^{\alpha}q^{2\alpha+\nu-2}x^{\alpha+\nu-1}}{\Gamma_{q^2}(1-\alpha )} D_{q,x} \int_x^\infty g_2(\rho) \rho^{1-2\alpha-\nu} ({x^2}/{\rho^2}; q^2)_{-\alpha}  \,d_q \rho.$$

By the above argument, If we assume that $\psi_1$ and $\psi_2$ are given by \eqref{13} and \eqref{131}, then
\begin{equation}\label{eq:1}\int_{0}^{\infty}u^{\alpha}C_1(u)J_{\nu-\alpha}(ux;q^2)\,d_qx=\left\{\begin{array}{cc}\phi_1(x),&x\in A_{q,b},\\ \psi_1(x),&x\in B_{q,b},\end{array}\right.\end{equation}
and
\begin{equation}\label{eq:2} \int_{0}^{\infty}u^{\alpha}C_2(u)J_{\nu+\alpha}(ux;q^2)\,d_qx=\left\{\begin{array}{cc}\phi_2(x),&x\in B_{q,a},\\ \psi_2(x),&x\in A_{q,a}.\end{array}\right.\end{equation}
Hence, \eqref{141} and  \eqref{14} follow by applying the  inverse pair of $q$-Hankel transforms~\eqref{HT} on \eqref{eq:1} and \eqref{eq:2}. This completes the proof.
\end{proof}

\begin{rem}\label{rem1}
From  the definitions of $\psi_i$ and $\phi_i$, $i=1,2$, in Proposition \ref{prop:1}, one can verify that $x^{-\nu-\alpha}\phi_2$ is bounded function in $B_{q,a}$ and and $x^{-\nu-\alpha}\psi_2$ is bounded in $A_{q,a}$. Also, $x^{-\nu+\alpha}\phi_1$ is bounded in $A_{q,b}$ and $x^{-\nu+\alpha}\psi_1$ is bounded in $B_{q,b}$.
\end{rem}

\begin{prop}\label{prop.2}
For $\rho\in B_{q,b}$, $\Psi_1(\rho)$ satisfies the Fredholm $q$-integral equation of the form
\begin{equation}\label{Fred.1}
\psi_1(\rho)=\tilde{F}_1(\rho)+\frac{q^{-2\alpha^2-\alpha+\nu}}{(1-q)^{2}}\int_{b}^{\infty} x \psi_1(x) K_1(\rho,x)\,d_qx,
\end{equation}
where
$$ K_1(\rho,x)= \int_0^\infty \frac{uw(u)}{1+w(u)} J_{\nu-\alpha}(ux;q^2) J_{\nu-\alpha}(u\rho;q^2)\, d_qu, $$
\[\begin{gathered}\tilde{F}_1(\rho)=  F_1(\rho)\,-\\
\frac{q^{-2\alpha^2-\alpha+\nu}}{(1-q)^{2}} \int_{0}^{a} x\,\psi_2(x) \int_{0}^{\infty} \dfrac{u}{1+w(u)}J_{\nu+\alpha}(ux;q^2) J_{\nu-\alpha}(u\rho ;q^2)\,d_qu\,d_qx,
\end{gathered}
\]
and
 \[\begin{gathered}F_1(\rho)=\rho^{\nu-\alpha}\, \frac{q^{-2\alpha^2-\alpha+\nu}(1+q)(1-q^2)^{-\alpha}}{(1-q)^2 \Gamma_{q^2}(\alpha )} \int_{\rho}^\infty x^{2\alpha-\nu-1}f_3(qx) (\rho^2/x^2; q^2)_{\alpha-1} \,d_qx \, -\\ \frac{q^{-2\alpha^2-\alpha+\nu}}{(1-q)^{2}}\Big[\int_a^\infty x\Phi_2(x) \int_0^\infty \dfrac{u}{1+w(u)} J_{\nu+\alpha}(ux;q^2)\ J_{\nu-\alpha}(u\rho;q^2)\,d_qu \,d_qx\\
  + \int_0^b x \Phi_1(x)\int_0^\infty \dfrac{uw(u)}{1+w(u)} J_{\nu-\alpha}(ux;q^2) J_{\nu-\alpha}(u\rho;q^2)\,d_qu \,d_qx\Big].
 \end{gathered}
 \]
 \end{prop}

 \begin{proof}
 Equation (\ref{7}) can be written in the following form:
 \begin{equation}\label{151}
  \int_0^\infty  u^{2\alpha}C_1(u) J_\nu (u\rho; q^2) \,d_q u= G(\rho), \quad \rho\in B_{q,b},
 \end{equation}
 where
 \begin{equation}\label{G}
 \begin{gathered}
 G(\rho)= f_3 (\rho) - \int_0^\infty u^{2\alpha} C_2(u)\dfrac{1}{1+w(u)} J_\nu (u\rho; q^2) \,d_q u \\ + \int_0^\infty  u^{2\alpha} C_1(u)\dfrac{w(u)}{1+w(u)}J_\nu (u\rho; q^2) \,d_q u.
 \end{gathered}
 \end{equation}
By using equations (\ref{a}) and (\ref{151}), we get
\begin{equation}\label{161}
G(\rho)= -(1-q) \rho^{\nu-1} q^{\nu-1} D_{q,\rho}\,\rho^{1-\nu} \int_0^\infty  u^{2\alpha-1}C_1(u) J_{\nu-1} (u\rho q^{-1}; q^2) \,d_q u.
 \end{equation}
 Substituting the value of $C_1(u)$ from  (\ref{141}) into  (\ref{161}), we obtain
\begin{equation}\label{17}
\begin{gathered}
 D_{q,\rho}\,\rho^{1-\nu} \int_0^\infty  u^{\alpha} \left[\int_0^b x \Phi_1(x)J_{\nu-\alpha}(ux; q^2)\,d_qx + \int_b^\infty x \Psi_1(x)J_{\nu-\alpha}(ux; q^2)\,d_qx \right] \times \\ J_{\nu-1} (u\rho q^{-1}; q^2)   \,d_q u =  - \frac{\rho^{1-\nu}q^{1-\nu} G(\rho)}{(1-q)} , \quad \rho\in B_{q,b}.
 \end{gathered}
 \end{equation}
From  (\ref{c}), there exists $M>0$ such that
\[\left| J_{\nu-\alpha}(ux;q^2)\right| \leq M(ux)^{\Re(\nu-\alpha)} \; \mbox{ for all } u,x \in \mathbb{R}_{q^2,b,+},\]
From Remark~\ref{rem1}, we have
 \[\begin{gathered}\left| \int_0^\infty  u^{\alpha} \left[\int_0^b x \Phi_1(x)J_{\nu-\alpha}(ux; q^2)\,d_qx + \int_b^\infty x \Psi_1(x)J_{\nu-\alpha}(ux; q^2)\,d_qx \right]J_{\nu-1} (u\rho q^{-1}; q^2) \,d_q u  \right| \\ \leq   M \Big[ \Big\| \Psi_1(x)\Big\|_{A_{q,b,\nu-\alpha}}+ \Big\| \Phi_1(x)\Big\|_{B_{q,b,\nu-\alpha}}\Big] \Big| \int_0^\infty u^{2\alpha+\nu}J_{\nu+1} (u\rho; q^2)\,d_qu \Big| <\infty .
 \end{gathered}
 \]
 Hence, the double $q$-integration is absolutely convergent and we can interchange the order of the $q$-integrations to obtain
\begin{equation}\label{181}
\begin{gathered}
G(\rho)= - (1-q)\rho^{\nu-1}q^{\nu -1} \left[\int_0^b x \Phi_1(x)\,d_qx + \int_b^\infty x \Psi_1(x)\,d_qx \right]\times \\ D_{q,\rho}\,\rho^{1-\nu} \int_0^\infty  u^{\alpha} J_{\nu-1} (u\rho q^{-1}; q^2) J_{\nu-\alpha}(ux; q^2) \,d_q u, \; \rho\in B_{q,b}.
\end{gathered}
\end{equation}
 Therefore, applying Proposition~\ref{cor} with $\Re(\nu-\alpha) > \Re(\nu-1) > -1$, we obtain
\begin{equation}\label{191}
G(\rho)= \frac{-(1-q)^2(1-q^2)^\alpha}{\Gamma_{q^2}(1-\alpha )} \rho^{\nu-1} D_{q,\rho} \int_{\rho}^\infty  x^{1-\nu-\alpha} \Psi_1(x)({\rho^2}/{x^2}; q^2)_{-\alpha} \,d_q x.
 \end{equation}
 By using
\begin{equation}\label{dq^2}
\int_x^\infty f(t) \,d_{q}t = \frac{1}{1+q} \int_{x^2}^\infty \frac{f(\sqrt{t})}{\sqrt{t}} \,d_{q^2}t,\;\;  D_{q,\rho}(f(\rho^2))= \rho (1+q) \left(D_{q^2}f\right)(\rho^2),
\end{equation}
 we obtain
\[\begin{split}G(\rho)&= \frac{-(1-q)^2(1-q^2)^\alpha }{\Gamma_{q^2}(1-\alpha )} \rho^{\nu} D_{q^2,\rho^2} \int_{\rho^2}^\infty  x^{\frac{-(\nu+\alpha)}{2}} \Psi_1(\sqrt{x})({\rho^2}/{x}; q^2)_{-\alpha} \,d_{q^2} x\\
&=-(1-q)^2 (1-q^2)^{\alpha}q^{\alpha^2-2\alpha-\nu}\rho^{\nu}\left(D_{q^2
} \mathcal{K}_{q^2}^{1-\alpha}\left((\cdot)^{-\frac{\nu+\alpha}{2}}\psi_1(\cdot)\right)\right)(\rho^2/q^2).
\end{split}
\]
Replacing $\rho$ by $q\rho$ yields
 \[-{q^{-\alpha^2+\alpha}(1-q^2)^{-\alpha}}{(1-q)^{-2}} \left[(\cdot)^{-\nu/2} G(q\sqrt{\cdot})\right](\rho^2)=  D_{q^2,\rho^2}\mathcal{K}_{q^2}^{1-\alpha}\left[(\sqrt{\cdot})^{(\alpha-\nu)}\psi_1(\sqrt{\cdot})\right](\rho^2).
 \]
Thus, applying Proposition \ref{prop.2} yields
\[\begin{gathered} \rho^{\alpha-\nu}\Psi_1(\rho) = q^{-\alpha^2}(1-q)^{-2}(1-q^2)^{-\alpha}\mathcal{K}_{q^2}^{\alpha}\left[(\cdot)^{-\nu/2}G(q\sqrt{\cdot})\right](\rho^2/q^2)\\
=\frac{q^{-2\alpha^2-\alpha+\nu}(1-q^2)^{-\alpha}(1-q)^{-2}}{\Gamma_{q^2}(\alpha)}\int_{\rho^2}^\infty  x^{-\frac{\nu}{2}+\alpha-1} G(q\sqrt{x})({\rho^2}/{ x}; q^2)_{\alpha-1 } \,d_{q^2} x.
\end{gathered}
\]
Using $\int_{x^2}^{\infty}f(t)\,d_{q^2}t=(1+q)\int_{x}^{\infty}tf(t^2)\,d_qt$, we obtain
 \[\rho^{\alpha-\nu}\Psi_1(\rho) =\frac{q^{-2\alpha^2-\alpha+\nu}(1-q^2)^{-\alpha}(1+q)}{(1-q)^{2}\Gamma_{q^2}(\alpha)}\int_{\rho}^\infty  x^{2\alpha-\nu-1} G(qx)({\rho^2}/{ x^2}; q^2)_{\alpha - 1} \,d_{q} x.\]
From (\ref{G}), we can write the last equation in the following form
\begin{equation}\label{gg}
\begin{gathered}
 \Psi_1(\rho)+ \rho^{\nu-\alpha}\, \frac{q^{-2\alpha^2-\alpha+\nu}(1-q^2)^{-\alpha}(1+q)}{(1-q)^{2}\Gamma_{q^2}(\alpha)}
  \int_\rho^{\infty} x^{2\alpha-\nu-1} (\rho^2/x^2; q^2)_{\alpha-1} \times \\ \Big[\int_0^\infty   \dfrac{u^{2\alpha}}{1+w(u)}C_2(u) J_\nu (qux; q^2) \,d_q u -  \int_0^\infty \dfrac{ w(u)}{1+w(u)}u^{2\alpha} C_1(u) J_\nu (qux; q^2) \,d_q u \Big] \,d_{q} x= \\ \rho^{\nu-\alpha}\,\frac{q^{-2\alpha^2-\alpha+
  \nu}(1-q^2)^{-\alpha}(1+q)}{(1-q)^2 \Gamma_{q^2}(\alpha )} \,  \int_{\rho}^\infty x^{2\alpha-\nu-1}\,f_3(qx) ({\rho^2}/{x^2}; q^2)_{\alpha-1} \,d_qx, \;\rho\in B_{q,b}.
 \end{gathered}
 \end{equation}
 From the condition on the function  $ C_2$, we can prove that the double $q$-integration
\begin{equation*}
 \int_\rho^{\infty} x^{2\alpha-\nu-1} (\rho^2/x^2; q^2)_{\alpha-1} \int_0^\infty   C_2(u)\dfrac{u^{2\alpha}}{1+w(u)} J_\nu (qux; q^2) \,d_q u\,d_qx
\end{equation*}
is absolutely convergent.  Therefore, we can interchange the order of the $q$-integrations and use Proposition \ref{int} to obtain
\begin{equation}\label{...}
\begin{gathered}
\Psi_1(\rho)+ \frac{q^{-2\alpha^2-\alpha+\nu}}{(1-q)^2}
\Big[\int_0^\infty   \dfrac{u^{\alpha}}{1+w(u)}C_2(u) J_{\nu-\alpha} (u\rho; q^2) \,d_q u - \\ \int_0^\infty \dfrac{u^{\alpha} w(u)}{1+w(u)} C_1(u) J_{\nu-\alpha} (u\rho; q^2) \,d_q u \Big] = \\ \rho^{\nu-\alpha}\,\frac{q^{-2\alpha^2-\alpha+\nu}(1-q^2)^{-\alpha}(1+q)}{(1-q)^2 \Gamma_{q^2}(\alpha )} \,  \int_{\rho}^\infty x^{2\alpha-\nu-1}\,f_3(qx) ({\rho^2}/{x^2}; q^2)_{\alpha-1} \,d_qx, \;\rho\in B_{q,b}.
 \end{gathered}
 \end{equation}
 Substitute the value of $C_1(u)$ and $C_2(u)$ from equations (\ref{14}) and (\ref{141}) into equation (\ref{...}), and then interchange the order of the $q$-integrations we get
\begin{equation}\label{Psi}
\begin{gathered}
 \Psi_1(\rho)+ \frac{q^{\nu-4\alpha}}{(1-q)^2}
 \Big[ \int_0^a x \psi_2(x) \int_{0}^{\infty} \dfrac{u}{1+w(u)}J_{\nu+\alpha}(ux;q^2) J_{\nu-\alpha}(u\rho ;q^2)\,d_qu\,d_qx \\- \int_b^\infty x\psi_1(x) \int_{0}^{\infty} \dfrac{u w(u)}{1+w(u)} J_{\nu-\alpha}(ux;q^2) J_{\nu-\alpha}(u\rho ;q^2)\,d_qu\,d_qx \Big] = F_1(\rho), \;\rho\in B_{q,b}.
 \end{gathered}
 \end{equation}
where
 \[\begin{gathered}F_1(\rho)=\rho^{\nu+\alpha}\, \frac{q^{\nu-4\alpha}(1+q)(1-q^2)^{-\alpha}}{(1-q)^2 \Gamma_{q^2}(\alpha )} \int_{\rho}^\infty x^{2\alpha-\nu-1}f_3(qx) (\rho^2/x^2; q^2)_{\alpha-1} \,d_qx -\\ \frac{q^{\nu-4\alpha}}{(1-q)^{2}}\Big[\int_a^\infty x\Phi_2(x) \int_0^\infty \dfrac{u}{1+w(u)} J_{\nu+\alpha}(ux;q^2)\ J_{\nu-\alpha}(u\rho;q^2)\,d_qu \,d_qx\\
  + \int_0^b x \Phi_1(x)\int_0^\infty \dfrac{uw(u)}{1+w(u)} J_{\nu-\alpha}(ux;q^2) J_{\nu-\alpha}(u\rho;q^2)\,d_qu \,d_qx\Big].
 \end{gathered}
 \]
 Equation (\ref{Psi}) is nothing else but the  Fredholm $q$-integral equation of the second kind~\eqref{Fred.1}. This completes the proof. \end{proof}

\begin{prop}\label{prop.3}
For $\rho\in A_{q,a}$, $\Psi_2(\rho)$ satisfies the Fredholm $q$-integral equation of the form
\begin{equation}\label{Fred.2}
\psi_2(\rho)=\tilde{F}_2(\rho)+\frac{1}{(1-q)^2}\int_{0}^{a} x K_2(\rho,x)\psi_2(x)\,d_qx,
\end{equation}
where
\[K_2(\rho,x)=
\int_0^\infty  \dfrac{uw(u)}{1+w(u)}  J_{\nu+\alpha}(ux; q^2) J_{\nu+\alpha} (u\rho; q^2)\,{d_qu},\]
\[\begin{gathered}\tilde{F}_2(\rho)=F_2(\rho)\, -\\
\frac{1}{(1-q)^2}\int_{b}^{\infty}x\Psi_1(x)\int_{0}^{\infty}\dfrac{u}{1+w(u)}
J_{\nu-\alpha}(ux;q^2)J_{\nu+\alpha}(u\rho;q^2)\,d_qu\,d_qx
\end{gathered}
\]
and
\begin{equation*}
\begin{gathered}
 F_2(\rho)= \frac{(1-q^2)^{-\alpha}(1+q) \rho^{\alpha-\nu-2}}{(1-q)^2\Gamma_{q^2}(\alpha )} \int_0^{\rho} ({q^2x^2}/{\rho^2}; q^2)_{\alpha-1} x^{\nu+1}f_1(x)\,d_qx \, + \\ \frac{1}{(1-q)^2} \int_a^\infty  x \Phi_2(x) \int_0^\infty  \dfrac{uw(u)}{1+w(u)}  J_{\nu+\alpha}(ux; q^2) J_{\nu+\alpha} (u\rho; q^2)\,{d_qu} \,{d_qx}\, -\\ \frac{1}{(1-q)^2} \int_0^b  x \Phi_1(x)\int_0^\infty \dfrac{u}{1+w(u)}J_{\nu-\alpha}(ux; q^2) J_{\nu+\alpha} (u\rho; q^2)\,{d_qu}\,{d_qx}.
 \end{gathered}
 \end{equation*}

\end{prop}

\medskip

\begin{proof}The proof is similar to the proof of Proposition~\ref{prop.2} and is omitted.

 \end{proof}

\begin{thm}\label{Thm:1}
The solution of \eqref{F:1}--\eqref{F:2} is given by
\[\Psi (u) =\dfrac{u^{2\alpha}}{1+w(u)}\left(C_1(u)+C_2(u)\right).\]
The functions $C_1$, $C_2$, $\phi_1$ and $\phi_2$ are given by Proposition \ref{prop:1}, and $\psi_1$, $\psi_2$ satisfies the Fredholm $q$-integral equations \eqref{Fred.2} and \eqref{Fred.1} of second kind.
\end{thm}
\vskip 0.5 cm

\noindent{\bf Example 1}

\noindent {\bf 1.} Take $b=aq^{-m}$ and assume that $m\to\infty$. If we assume that $f_1=f$, $f_2=f$, and $w=0$. Then  the system \eqref{F:1}--\eqref{F:3} is reduced to
the dual $q$-integral equations
\begin{eqnarray}
\int_{0}^{\infty}\psi(u)J_{\nu}(u\rho;q^2)\,d_qu=f(\rho),\quad \rho\in A_{q,a}\\
\int_{0}^{\infty}u^{-2\alpha}\psi(u)J_{\nu}(u\rho;q^2)\,d_qu=0,\quad \rho\in B_{q,a}.
\end{eqnarray}
 Hence, from Theorem~\ref{Thm:1}
 \[\psi(u)=u^{1+\alpha}\int_{0}^{\infty}x\psi_2(x)J_{\nu+\alpha}(ux;q^2)\,d_qx,\;u\in\mathbb{R}_{q,+}\]
 \[\begin{split}\psi_2(\rho)&=\dfrac{(1-q^2)^{-\alpha}(1+q)\rho^{\alpha-\nu-2}}{(1-q)^2\Gamma_{q^2}(\alpha)}\int_{0}^{\rho}(q^2 x^2/\rho^2;q^2)_{\alpha-1}x^{\nu+1}f(x)\,d_qx\\
 &=\rho^{-\alpha-\nu}\dfrac{(1-q^2)^{-\alpha}}{(1-q)^2}I_{q^2}^{\alpha}\left(t^{\nu/2}f(\sqrt{t})\right)(\rho^2).
 \end{split}\]
 Hence,
 \[\psi(u)=u^{1+\alpha}\dfrac{(1-q^2)^{-\alpha}}{(1-q^2)}\int_{0}^{\infty}x^{1-\alpha-\nu}I_{q^2}^{\alpha}\left(t^{\nu/2}f(\sqrt{t})\right)(x^2)J_{\nu+
 \alpha}(ux;q^2)\,d_qx.\]
This coincides with the result in~\cite[Theorem 4.1]{O.A} for solutions of double $q$-integral equations.
\vskip 0.5 cm
Let $a=q^m$ and assume that $m\to\infty$. If we assume that $f_2=0$, and  $f_3=f$, we obtain the dual $q$-integral system of equations
\begin{eqnarray}
\int_{0}^{\infty}u^{-2\alpha}\psi(u)J_{\nu}(u\rho;q^2)\,d_qu=0,\quad \rho\in A_{q,b}\\
\int_{0}^{\infty}\psi(u)J_{\nu}(u\rho;q^2)\,d_qu=f,\quad \rho\in B_{q,b}.
\end{eqnarray}
 Hence, from Theorem~\ref{Thm:1}
 \[\psi(u)=u^{1+\alpha}\int_{b}^{\infty}x\psi_1(x)J_{\nu-\alpha}(ux;q^2)\,d_qx,\;u\in \mathbb{R}_{q,+},\]
 \[\psi_1(\rho)=-\dfrac{(1-q^2)^{-\alpha}q^{-2\alpha}\rho^{\alpha+\nu}}{(1-q)^2\Gamma_{q^2}(\alpha)}\int_{\rho}^{\infty}(\rho^2/ x^2;q^2)_{\alpha-1}x^{2\alpha-\nu-1}f(x)\,d_qx.\]
This is a special case of~\cite[Theorem 5.1]{O.A}.
\vskip 0.5 cm

\noindent{\bf Example 2}

We consider the triple $q$-integral equations
\begin{eqnarray}
\int_{0}^{\infty}\psi(u)J_{0}(u\rho;q^2)\,d_qu&=&0,\quad \rho\in A_{q,a},\\
\int_{0}^{\infty}u^{-1}\psi(u)J_{0}(u\rho;q^2)\,d_qu&=&1,\quad \rho\in A_{q,b}\cap B_{q,a},\\
\int_{0}^{\infty}\psi(u)J_{0}(u\rho;q^2)\,d_qu&=&0,\quad \rho\in B_{q,b}.
\end{eqnarray}
Hence, we have $\nu=0$, $g_1=1$, $g_2=0$, $f_1=f_3=0$, $w=0$, and $\alpha=\frac{1}{2}$.

From Theorem \ref{Thm:1},
\[\psi(u)=u\left(C_1(u)+C_2(u)\right),\]
where
\[C_1(u)=\frac{(1-q)(1-q^2)}{\Gamma_{q^2}^2({1}/{2})}\dfrac{\sin(\frac{bu}{1-q};q)}{u}+\frac{\sqrt{1-q^2}}{\Gamma_{q^2}(1/2)}\int_{b}^{\infty}\sqrt{x}\psi_1(x)\cos(\frac{xu\sqrt{q}}{1-q};q^2)\,d_qx,\]
\[C_2(u)=\frac{\sqrt{1-q^2}}{\Gamma_{q^2}(1/2)}\int_{0}^{a}\sqrt{x}\psi_2(x)\sin(\frac{xu}{1-q};q^2)\,d_qx,\]
\be \label{Eq:1-2}\psi_1(\rho)=\frac{\sqrt{\rho}(1+q)}{q(1-q)\Gamma_{q^2}^2(1/2)}\int_{0}^{a}x^{3/2}\dfrac{\psi_2(x)}{q\rho^2-x^2}\,d_qx,\; \rho\in B_{q,b},\ee
\be\begin{gathered} \label{Eq:2-2}\psi_2(\rho)=-\dfrac{(1+q)\sqrt{\rho}}{(1-q)\Gamma_{q^2}^2(1/2)}\int_{b}^{\infty}\dfrac{\sqrt{x}\psi_1(x)}{qx^2-\rho^2}\,d_qx\\+
\dfrac{(1+q)^{3/2}}{\sqrt{1-q}\Gamma_{q^2}^{3}(1/2)}\sqrt{\rho}\int_{\rho/q}^{b}\dfrac{d_qx}{qx^2-\rho^2}.
\end{gathered}
\ee
We used \cite[PP. 455-466]{H.T} or \cite[Proposition 2.4 ]{O.A} to calculate $\psi_1$ and $\psi_2$ in equations \eqref{Eq:1-2} and \eqref{Eq:2-2}, respectively. Substituting from \eqref{Eq:1} into \eqref{Eq:2}, we obtain the second order Fredholm $q$-integral equation
\be
\psi_2(\rho)=-\frac{q^{-1}\sqrt{\rho}(1+q)}{(1-q)^2\Gamma_{q^2}^3(1/2)}\int_{0}^{a}t^{3/2}\psi_2(t)K_2(\rho,t)\,d_qt+
\dfrac{(1+q)^{3/2}}{\sqrt{1-q}\Gamma_{q^2}^{3}(1/2)}\sqrt{\rho}\int_{\rho/q}^{b}\dfrac{d_qx}{qx^2-\rho^2}.,
\ee
where $\rho\in A_{q,a}$ and
\[K(\rho,t)=\int_{b}^{\infty}\frac{x}{(t^2-qx^2)(\rho^2- q x^2)}\,d_qt.\]
\section{{\bf Solving  system of triple $q^2$-Integral Equations by using solutions of dual $q$-integral equations}}

In~\cite{J.C.1}, Cooke solved  certain triple integral equations involving Bessel functions by using a result for Noble~\cite{Nob1} for solutions for dual integral equations with Bessel functions as kernel.  In this section, we use  the result, Theorem A , which  introduced in~\cite{O.A} to solve the following triple q-integral equations:
\begin{equation}\label{E:1}
\xi^{-\gamma} \int_0^\infty  \rho^{-\gamma} \psi (\rho) J_\kappa (\sqrt{\rho \xi}; q^2) \,d_{q^2}\rho= f (\xi), \quad   \xi\in A_{q^2}
 \end{equation}
 \begin{equation}\label{E:2}
\xi^{-\alpha} \int_0^\infty  \rho^{-\alpha} \psi (\rho) J_\mu (\sqrt{\rho \xi}; q^2) \,d_{q^2}\rho= g (\xi), \quad   \xi\in A_{q^2}\cap B_{q^2}
  \end{equation}
  \begin{equation}\label{E:3}
\xi^{-\beta} \int_0^\infty  \rho^{-\beta} \psi (\rho) J_\nu (\sqrt{\rho \xi}; q^2) \,d_{q^2}\rho= h (\xi), \quad   \xi\in B_{q^2}
\end{equation}
where  $a$, $\alpha$, $\beta$, $\gamma$, $\mu$, $\nu$ and $\kappa$  are  complex numbers such that \[\Re(\nu) > -1,\;\; \Re(\mu)>-1,\;\; \Re(\kappa)>-1,\;\mbox{and}\; 0<a<1,\] , the functions
$f (\rho)$, $g (\rho)$ and $h (\rho)$ are known functions, and  $\psi (u)$ is the solution function  to be determined.

The following is a result from \cite{O.A} that we shall use to solve the system \eqref{E:1}--\eqref{E:3}.

\vskip .5 cm
\noindent{\bf Theorem A.}
{\it Let $\alpha$, $\beta$, $\mu$ and $\nu$ be complex numbers and let $\lambda:=\frac{1}{2}(\mu+\nu)-(\alpha-\beta)> -1$. Assume that $$ \Re(\nu)> -1,\, \Re(\mu)> -1,\, \Re(\lambda)> -1,\; \mbox { and } \Re(\lambda-\mu-2\alpha)>0.$$ Let $f\in L_{q^2,\frac{\mu}{2}+\alpha }(A_{q^2})$ and $g\in L_{q^2,\frac{-\mu}{2}+\alpha -1 }(B_{q^2})$. Then the dual $q^2$-integral equations
\begin{equation}\label{dual.1}
 \xi^{-\alpha} \int_0^\infty  \rho^{-\alpha} \psi (\rho) J_\mu (\sqrt{\rho \xi}; q^2) \,d_{q^2}\rho= f(\xi), \quad   \xi\in A_{q^2},
 \end{equation}
 \begin{equation}\label{dual.2}
 \xi^{-\beta} \int_0^\infty  \rho^{-\beta} \psi (\rho) J_\nu (\sqrt{\rho \xi}; q^2) \,d_{q^2}\rho= g(\xi), \quad  \xi\in B_{q^2}
\end{equation}
 has the solution of the form
\[
\begin{gathered}
\psi(\xi)= (1-q^2)^{\lambda-\nu+2\alpha-2}\xi^{\lambda/2-\mu/2+\alpha} \int_0^1 J_{\lambda} (\sqrt{\rho \xi}; q^2) I^{\mu/2+\alpha,\lambda-\mu}_{q^2}f(\rho)\,d_{q^2}\rho \\ + (1-q^2)^{\lambda-\nu-2}\xi^{\lambda/2-\mu/2+\alpha} \int_1^\infty J_{\lambda} (\sqrt{\rho \xi}; q^2) \mathcal{K}^{\lambda/2-\nu/2-\beta,\nu-\lambda}_{q^2}g(\rho)\,d_{q^2}\rho,
\end{gathered}
\]
in $L_{q^2,\frac{\mu}{2}-\alpha }(\mathbb{R}_{q^2,+})\cap L_{q^2,\frac{\nu}{2}-\beta }(\mathbb{R}_{q^2,+})\cap L_{q^2,\frac{\nu}{2}-\beta-\gamma }(\mathbb{R}_{q^2,+}) $, for $\gamma$ satisfying
\[ 1+\Re(\nu)>\Re(\gamma)> \mbox{ max } \{0, \Re(\nu-\lambda)\}.\]
}
\vskip 0.2 cm

Now, we shall solve the system of triple $q^2$-integral equations \eqref{E:1}--\eqref{E:3}.  Since the function $g(\rho)$ is only defined in $A_{q^2}\cap B_{q^2}$, we can write $$ g(\xi)=g_1(\xi)+g_2(\xi),$$ $g_1$ and $g_2$ defined in $A_{q^2}$ and $B_{q^2}$ respectively. So, we may assume that $$ \psi= A_1+A_2, $$ and we solve the equations in the form

\begin{equation}\label{I_1}
 \xi^{-\gamma} \int_0^\infty  \rho^{-\gamma} [A_1(\rho)+A_2(\rho)] J_\kappa (\sqrt{\rho \xi}; q^2) \,d_{q^2}\rho= f (\xi), \quad \xi\in A_{q^2},
 \end{equation}

 \begin{equation}\label{I_2}
 \xi^{-\alpha} \int_0^\infty  \rho^{-\alpha} A_1(\rho) J_\mu (\sqrt{\rho \xi}; q^2) \,d_{q^2}\rho = g_1 (\xi), \quad   \xi\in A_{q^2},
  \end{equation}

 \begin{equation}\label{I_3}
 \xi^{-\alpha} \int_0^\infty  \rho^{-\alpha} A_2(\rho) J_\mu (\sqrt{\rho \xi}; q^2) \,d_{q^2}\rho = g_2(\xi), \quad   \xi\in B_{q^2},
  \end{equation}

\begin{equation}\label{I_4}
\xi^{-\beta} \int_0^\infty  \rho^{-\beta} [A_1(\rho)+A_2(\rho)] J_\nu (\sqrt{\rho \xi}; q^2) \,d_{q^2}\rho= h (\xi), \quad  \xi\in B_{q^2},
 \end{equation}
 We rewrite the equations as two pairs of dual $q$-integral equations, namely
\begin{equation}\label{I_5}
 \left\{
\begin{array}{ll}
 \displaystyle
  \xi^{-\alpha} \int_0^\infty  \rho^{-\alpha} A_1(\rho) J_\mu (\sqrt{\rho \xi}; q^2) \,d_{q^2}\rho = g_1 (\xi), \quad   \xi\in A_{q^2}, \cr\cr
 \displaystyle \xi^{-\beta} \int_0^\infty  \rho^{-\beta} A_1(\rho) J_\nu (\sqrt{\rho \xi}; q^2) \,d_{q^2}\rho  = h (\xi)- f_2(\xi)  , \quad   \xi\in B_{q^2}
\end{array}
         \right.
\end{equation}

\begin{equation}\label{I_6}
 \left\{
\begin{array}{ll}
 \displaystyle
  \xi^{-\alpha} \int_0^\infty  \rho^{-\alpha} A_2(\rho) J_\mu (\sqrt{\rho \xi}; q^2) d_{q^2}\rho = g_2(\xi), \quad   \xi\in B_{q^2},    \cr\cr
 \displaystyle \xi^{-\gamma} \int_0^\infty  \rho^{-\gamma} A_2(\rho) J_\kappa (\sqrt{\rho \xi}; q^2) \,d_{q^2}\rho= f (\xi)- f_1(\xi), \quad   \xi\in A_{q^2},
\end{array}
         \right.
\end{equation}
 where $$ \xi^{-\gamma} \int_0^\infty  \rho^{-\gamma} A_1(\rho) J_\kappa (\sqrt{\rho \xi}; q^2) d_{q^2}\rho  =f_1(\xi), \quad\quad \xi\in A_{q^2},$$
$$  \xi^{-\beta} \int_0^\infty  \rho^{-\beta} A_2(\rho) J_\nu (\sqrt{\rho \xi}; q^2) \,d_{q^2}\rho   =f_2(\xi), \quad\quad \xi\in B_{q^2},$$
 Then we can solve the first and second pairs by Theorem \ref{dual theorem}. For the first pairs
\begin{equation*}
A_1(\xi)= (1-q^2)^{\lambda-\nu+2\alpha-2}\xi^{\lambda/2-\mu/2+\alpha}  \int_0^1 J_{\lambda} (\sqrt{\rho \xi}; q^2) I^{\mu/2+\alpha,\lambda-\mu}_{q^2} g_1(\rho)\,d_{q^2}\rho $$$$ +  (1-q^2)^{\lambda-\nu-2}\xi^{\lambda/2-\mu/2+\alpha} \int_1^\infty J_{\lambda} (\sqrt{\rho \xi}; q^2) \mathcal{K}^{\lambda/2-\nu/2-\beta,\nu-\lambda}_{q^2}[h(\rho)-f_2(\rho)]\,d_{q^2}\rho,
\end{equation*}
where, $\lambda:=\frac{1}{2}(\mu+\nu)-(\alpha-\beta)> -1$.
 The solution of the second pair has the form
\begin{equation*}
A_2(\xi)= (1-q^2)^{\lambda-\mu+2\gamma-2}\xi^{\lambda/2-\kappa/2+\gamma}  \int_0^a J_{\lambda} (\sqrt{\rho \xi}; q^2) I^{\kappa/2+\gamma,\lambda-\kappa}_{q^2} [f(\rho)-f_1(\rho)] \,d_{q^2}\rho $$$$ +  (1-q^2)^{\lambda-\mu-2}\xi^{\lambda/2-\kappa/2+\gamma} \int_a^\infty J_{\lambda} (\sqrt{\rho \xi}; q^2) \mathcal{K}^{\lambda/2-\mu/2-\alpha,\mu-\lambda}_{q^2}  g_2(\rho)d_{q^2}\rho,
\end{equation*}
where, $\lambda:=\frac{1}{2}(\mu+\kappa)-(\gamma-\alpha)> -1$.

\bibliographystyle{plain}

\end{document}